\pdfoutput=1
\documentclass[a4paper,runningheads]{llncs}
\title{Linear Recurrence Sequence Automata and the Addition of Abstract Numeration Systems\thanks{A preliminary version of this work was presented under the title \emph{Addition in Dumont-Thomas Numeration Systems in Theory and Practice} in the 2024 edition of \emph{Journ\'ees Montoises d'Informatique Th\'eorique}.}}
\titlerunning{Linear Recurrence Sequence Automata}

\author{Olivier Carton\inst{1}\orcidlink{0000-0002-2728-6534} \and Jean-Michel Couvreur\inst{2}\orcidlink{0000-0002-1006-7132} \and\newline Martin Delacourt\inst{2}\orcidlink{0000-0001-5158-4606} \and Nicolas Ollinger\inst{2}\orcidlink{0000-0002-6826-1034}\thanks{Corresponding author}}
\institute{Université Paris Cité, CNRS, IRIF, F-75013, Paris, France
\and 
Université d'Orléans, INSA CVL, LIFO, UR 4022, Orléans, France}
\authorrunning{O. Carton et al.}

\usepackage{microtype,amssymb,amsmath,graphicx,tikz,hyperref,stackengine,mathcommand,enumerate,todonotes,xspace,tikzlings,orcidlink,nicematrix}
\usepackage[labelfont=rm]{subcaption}

\usetikzlibrary{calc,decorations.pathreplacing,shapes.misc}
\tikzset{
    show control points/.style={
        decoration={show path construction, curveto code={
                \draw [blue, dashed]
                    (\tikzinputsegmentfirst) -- (\tikzinputsegmentsupporta)
                    node [at end, cross out, draw, solid, red, inner sep=2pt]{};
                \draw [blue, dashed]
                    (\tikzinputsegmentsupportb) -- (\tikzinputsegmentlast)
                    node [at start, cross out, draw, solid, red, inner sep=2pt]{};
            }
        },
        postaction=decorate
    },
}
\usepackage{pict2e}
\makeatletter
\DeclareRobustCommand{\dplus}{\DOTSB\mathop{\dplus@}\slimits@}
\newcommand{\dplus@}{\vphantom{\sum}\mathpalette\dplus@@\relax}
\newcommand{\dplus@@}[2]{%
  \begingroup
  \sbox\z@{$#1\sum$}%
  \unitlength=\dimexpr\ht\z@+\dp\z@\relax
  \linethickness{%
    \ifx#1\displaystyle 1.8\fontdimen8\textfont3 \else
    \ifx#1\textstyle 1.2\fontdimen8\textfont3 \else
    \ifx#1\scriptstyle 1.2\fontdimen8\scriptfont3 \else
    1.3\fontdimen8\scriptscriptfont3 \fi\fi\fi}
  \vcenter{\hbox{%
    \begin{picture}(0.8,0.8)
    \polygon(0.04,0.4)(0.4,0.04)(0.76,0.4)(0.4,0.76)
    \Line(0.4,0.04)(0.4,0.76)
    \Line(0.04,0.4)(0.76,0.4)
    \end{picture}%
  }}%
  \endgroup
}
\makeatother

\def\squareforqed{\hbox{$\blacksquare$}}
\newcommand\NN{\mathbb{N}}
\newcommand\Acal{\mathcal{A}}
\newcommand\Ncal{\mathcal{S}}
\newcommand\KK{\mathbb{K}}
\newcommand\ZZ{\mathbb{Z}}
\newcommand\Zee{\mathbb{Z}}
\newcommand\CC{\mathbb{C}}
\newcommand\RR{\mathbb{R}}
\newcommand\QQ{\mathbb{Q}}
\newcommand\addrel{L^+}
\newcommand\zseq{(0)}
\newcommand\sequ{{\bf u}}
\newcommand\seqv{{\bf v}}
\newcommand\lrec{E}
\newcommand\sshift{\sigma}
\newcommand\val{\operatorname{val}}
\newcommand\rep{\operatorname{rep}}
\newcommand\supp{\operatorname{supp}}
\usetikzlibrary{automata,positioning}
\tikzset{every state/.style={minimum size=0pt}}
\newcommand{\seqnum}[1]{\href{https://oeis.org/#1}{\rm \underline{#1}}}

\begin{document}
\maketitle

\begin{abstract}
Abstract numeration systems encode natural numbers using radix ordered words of an infinite regular language and linear recurrence sequences play a key role in their valuation. Sequence automata, which are deterministic finite automata with an additional linear recurrence sequence on each transition, are introduced to compute various $\ZZ$-rational non commutative formal series in abstract numeration systems. Under certain Pisot conditions on the recurrence sequences, the support of these series is regular. This property can be leveraged to derive various synchronized relations including a deterministic finite automaton that computes the addition relation of various Dumont-Thomas numeration systems and deterministic finite automata converting between various numeration systems. A practical implementation for Walnut is provided.
\end{abstract}

Calculators are handy tools to accelerate tedious computations. A particular type of computation arise from the interplay between formal languages, automata and logic initiated by Büchi \cite{buchi} and described in the comprehensive survey by Bruyère \textit{et al.}~\cite{bruysurvey}. Practical tools implement these ideas, such as the Walnut proof assistant developed by Mousavi, Shallit \textit{et al.}~\cite{walnuta,walnutb}.
Walnut particularly shines at the task to check first-order arithmetical properties on automatic sequences,
provided the computational complexity is tractable on the examples considered. Indeed, the decision procedure is \textsc{Tower}-complete for these arithmetics \cite{schmitz}. 

The present work is motivated by the study of first-order properties of infinite sequences constructed as fixpoints $\tau^\omega(a) = \lim_{k \to\infty}\tau^k(a)$ of some non-uniform substitutions, such as ${\tau : a\mapsto aab,\; b\mapsto cab,\; c\mapsto a}$ for which
$\tau^\omega(a) = aabaabcabaabaabcabaaabcabaabaabcabaabaabcabaaabcabaabaabaab\dots$.
To answer a question on this sequence in a mechanized way, such as \emph{What are the possible lengths of words appearing as cubes inside $\tau^\omega(a)$?}, one could express the question as a first-order predicate over $\left<\NN, +, \leqslant, \tau^\omega(a)\right>$:
\begin{align*}
    \operatorname{FactorEq}(i,j,n) &\equiv \forall k\; k<n \rightarrow \tau^\omega(a)[{i+k}] = \tau^\omega(a)[{j+k}]\\
    \operatorname{Cubes}(n) &\equiv n\geqslant 1 \land \exists i\; \operatorname{FactorEq}(i, i+n, n) \land \operatorname{FactorEq}(i, i+2n, n)
\end{align*}
A required step to mechanize these computations is to identify a numeration system in which the infinite word $\tau^\omega(a)$ is automatic as well as the addition relation $\left\{ (x,y,z) \mid x+y=z \right\}$ and the order relation $\left\{ (x,y) \mid  x\leqslant y \right\}$. From there, the calculator computes the answer, here the possible length correspond to the linear recurrence sequence \cite[\seqnum{A098182}]{Sloane} given by initial values $1, 3, 9$ and recurrence relation $u_{n+3} = 3u_{n+2}-u_{n+1}+u_n$.

An immediate candidate abstract numeration system \cite{rigo2,cant3} to describe the fixpoint as an automatic word is the associated Dumont-Thomas numeration system \cite{dumontthomas,rigo2002,rigo2}. In this numeration system, the appearance of a symbol is explained by the address of that symbol, the path taken in the substitution tree generated by $\tau$ starting from $a$. The language of valid labels is recognized by the addressing automaton of $\tau^\omega(a)$ depicted on Fig.~\ref{fig:adtau}. The states of the addressing automaton are the symbols of the substitution and from a state starts one edge per letter in its image by the substitution, labelled by successive integers. The order relation in this system is the radix order which is always regular. What about the regularity of the language encoding the addition relation?
\begin{figure}[t]
\centering
\subcaptionbox{Addressing automaton $\Ncal_\tau$\label{fig:adtau}}[.35\textwidth]
{\begin{tikzpicture}[thick,shorten >=1pt,node distance=1.5cm,on grid,auto,initial text=] 
   \node[state,initial,accepting] (q_a) {$a$}; 
   \node[state,accepting] (q_b) [right=of q_a] {$b$}; 
   \node[state,accepting] (q_c) [right=of q_b] {$c$}; 
    \path[->] 
    (q_a) edge [loop below] node {0,1} ()
          edge [bend left] node {2} (q_b)
    (q_b) edge [loop above] node  {2} ()
          edge node {0} (q_c)
          edge node {1} (q_a)
    (q_c) edge [bend left=45] node {0} (q_a)
          ;
\end{tikzpicture}}
\subcaptionbox{Sequence automaton $\mathcal{A}_\tau$ (\scriptsize $P=X^3-3X^2+X-1$)\label{fig:seqpsi}}[.55\textwidth]
{\begin{tikzpicture}[thick,shorten >=1pt,node distance=2.2cm,on grid,auto,initial text=] 
   \node[state,initial,accepting] (q_a) {$a$}; 
   \node[state,accepting] (q_b) [right=of q_a] {$b$}; 
   \node[state,accepting] (q_c) [right=of q_b] {$c$}; 
    \path[->] 
    (q_a) edge [loop below] node[align=left] {\scriptsize $0:(0,0,0)$,\\\scriptsize $1:(1,3,9)$} ()
          edge [bend left=20] node {\scriptsize $2:(2,6,18)$} (q_b)
    (q_b) edge [loop above] node  {\scriptsize $2:(2,4,12)$} ()
          edge node {\scriptsize $0:(0,0,0)$} (q_c)
          edge node {\scriptsize $1:(1,1,3)$} (q_a)
    (q_c) edge [bend left=40] node {\scriptsize $0:(0,0,0)$} (q_a)
          ;
\end{tikzpicture}}
\caption{Addressing and sequence automata for $\tau : a\mapsto aab,\; b\mapsto cab,\; c\mapsto a$}
\end{figure}

To the best of the authors' knowledge, the regularity of addition has only been extensively studied in the context of positional numeration systems. Let $(\sequ_i)$ be an increasing sequence of positive integers starting from $\sequ_0=1$. Each number is represented greedily as $n=\sum_{i=0}^m \alpha_i \sequ_i$ where $\alpha_m\cdots\alpha_0$ is the representation of $n$ in the positional numeration system $\sequ$. When the sequence $\sequ$ is a linear recurrence sequence whose characteristic polynomial is the minimal polynomial of a Pisot number, the associated numeration system is a Pisot numeration system. Bruyère and Hansel \cite{bruhan} and Frougny and Solomyak \cite{frousol} have established the regularity of addition for Pisot numeration systems, including their canonical representative which are Bertrand numeration systems.

Dumont-Thomas numeration systems are typically neither positional nor greedy. However, given a current state in the addressing automaton the choice for the next symbol in the address is greedy among available transitions. In the literature, to study a fixpoint of a substitution, one would search for a Pisot numeration system whose minimal polynomial is the characteristic polynomial of the incidence matrix of the substitution and express the fixpoint as an automatic sequence in this positional numeration system. In the present paper, we leverage the techniques used for Pisot numeration systems to obtain a regularity of the language of the addition relation for some abstract numeration systems. In contrast to the positional numeration systems case where the normalization~\cite{frounor} relations are convenient relations to study in order to obtain results on the addition relation, there is no clear equivalent for abstract numeration systems.

\begin{theorem}\label{addDT}
The addition relation of the abstract numeration system associated with a Pisot substitution, a substitution with an incidence matrix whose characteristic polynomial is the minimal polynomial of a Pisot number, is synchronized.
\end{theorem}

\section{Addition relation for abstract numeration systems}

An \emph{abstract numeration system with zeros (ANSZ)} $\Ncal$ is a tuple $(L,A,<,0)$ where $(A,<)$ is an ordered alphabet of minimal element $0\in A$ and $L$ is a regular language over $A$ containing $\varepsilon$ and such that $w\in L \Leftrightarrow 0w\in L$ for all $w\in A^*$. It is a slight variation on the classical regular abstract numeration systems \cite{rigo2,cant3} where padding with 0 on the left is allowed \cite{marsault2017,charlier2022}.

The \emph{representation $\rep_\Ncal(n)$} of an integer $n\in\NN$ is defined by the canonical bijection between $\NN$ and $L\setminus 0^+ L$ ordered in radix order --- words are compared first by length and then by lexicographic order.

The \emph{valuation function $\val_\Ncal$} maps $L$ to $\NN$ so that $\val_\Ncal(w)$ is the only natural number $n\in\NN$ such that $w\in 0^*\rep_\Ncal(n)$.

\subsection{Automatic, synchronized and regular sequences}

Let  $\Ncal$ be an ANSZ $(L,A,<,0)$.
An infinite word ${\bf x}\in B^\NN$ is \emph{$\Ncal$-automatic} \cite{allouche2003,rigo2002} if there exist a deterministic finite automaton producing the output ${\bf x}[\val_\Ncal(w)]$ for all input word $w\in L$.

Let $\left<.,.\right>$ denote the canonical bijection between $\bigcup_{n\geqslant 0}\left(A^n\times A^{\prime n}\right)$ and ${(A\times A')^*}$ and extend it to tuples of arbitrary size by letting $\left<x,y,z\right> = \left<x,\left<y,z\right>\right>$, \textit{etc}.
A relation $R \subseteq \NN^m$ is \emph{$\Ncal$-synchronized} \cite{carpi2010} if
$\left\{ \left<y_1,\ldots,y_m\right> \mid (\val_\Ncal(y_1),\ldots,\val_\Ncal(y_m)) \in R \right\}$
is regular.
The \emph{addition relation} of an ANSZ $\Ncal=(L,A,<,0)$ is the language 
\[
\addrel_\Ncal = {\Bigl\{ \left<x,y,z\right> \;\Bigm|\; x,y,z\in L \land \val_\Ncal(x) + \val_\Ncal(y) = \val_\Ncal(z) \Bigr\}}\quad.
\]

A sequence of natural numbers ${\bf a} : \NN\rightarrow \NN$ is \emph{$\Ncal$-regular} \cite{allouche1992,rigo2002,charlier2022} if the function $f : A^*\rightarrow\NN$, defined by $f(w) = {\bf a}(\val_\Ncal(w))$ on $L$ and $f(w)=0$ otherwise, is a $\ZZ$-rational formal power series. That is, by Schützenberger's theorem \cite{berstel2011}, if it is recognizable by a \emph{linear representation} $(\lambda,\mu,\rho)$ consisting of a row vector $\lambda$, a matrix-valued morphism $\mu : A^* \rightarrow \ZZ^{m\times m}$ and a column vector $\rho$, such that ${\bf a}(n) = \lambda\mu(\rep_\Ncal(n))\rho$ for all $n\geqslant 0$ and $\lambda\mu(w)\rho=0$ if $w$ is not in $L$.

\subsection{ANSZ and Dumont-Thomas numeration systems}

Let $\Ncal$ and $\Ncal'$ be two ANSZ. The \emph{conversion relation} between $\Ncal$ and $\Ncal'$ is the language
${\left\{ \left<x,y\right> \mid \val_\Ncal(x) = \val_{\Ncal'}(y) \right\}}$.
Two ANSZ are said to be \emph{equivalent} if their conversion relation is a regular language.

\begin{proposition}
Equivalence between ANSZ preserve automaticity, synchronicity and regularity of sequences and relations.
\end{proposition}

A \emph{substitution} $\tau : A \rightarrow A^*$ generates a fixpoint $\tau^\omega(a) \in A^\NN$ from a letter $a\in A$ as $\tau^\omega(a) = a\prod_{n\in\NN}\tau^n(w)$ provided that $\tau(a)=aw$ for some $w\in A^+$ such that no $\tau^k(w)$ is the empty word.
Dumont and Thomas \cite{dumontthomas} proved that every prefix $p$ of $\tau^\omega(a)$ can be represented using a unique sequence $\left(p_i,a_i\right)_{i=0}^k\in\left(A^*\times A\right)^*$ as ${p = \prod_{i=0}^k \tau^{k-i}(p_{i})}$ where $p_{i+1} a_{i+1}$ is a prefix of $\tau(a_i)$ for all $i<k$ and $p_0 a_0$ is a prefix of $\tau(a)$. The sequence of the lengths of the $p_i$, the \emph{address} of letter $\tau^\omega(a)_{|p|}$, completely characterizes the prefix $p$. 

The \emph{addressing automaton} $\Ncal_\tau$ associated with the fixpoint $\tau^\omega(a)$ is the deterministic finite automaton with state set~$A$, alphabet $B=\{0,1,\ldots,n-1\}$ where $n=\max_{b\in A}|\tau(b)|$, initial state~$a$, final states~$A$ and whose transitions are defined by $\tau$ as $\delta(b,i) = \tau(b)_i$ for all $b\in A$ and $i \in \{0,\ldots,|\tau(b)|-1\}$. The \emph{Dumont-Thomas numeration system} associated with $\tau^\omega(a)$ is the ANSZ $(L(\Ncal_\tau), B, <, 0)$. The fixpoint $\tau^\omega(a)$ is automatic in this numeration system, it is recognized by $\Ncal_\tau$ where the output of a state $b\in A$ is the letter $b$.

By taking only a subset $F\subseteq A$ as accepting states for $\Ncal_\tau$, one defines the \emph{Generalized Dumont-Thomas numeration system} associated with $\pi(\tau^\omega(a))$ where $\pi$ is the erasing substitution mapping every element of $F$ to itself and all the other elements to $\varepsilon$. It relates with label reduction and surminimisation in \cite{marsault2015}.

\begin{proposition}\label{progdtns}
Every ANSZ is equivalent to a generalized Dumont-Thomas numeration system.
\end{proposition}

\subsection{$\NN$-rational valuation of ANSZ}\label{sec:val}

Let $\Ncal=(L,A,<,0)$ be an ANSZ with $A=\{0,\ldots,m-1\}$ and $<$ the classical order on $\NN$. Let $(Q,A,\delta,0,F)$ be a DFA recognizing $L$ with $Q=\{0,\ldots,n-1\}$ and $\delta : Q\times A \rightarrow Q$ its partial transition map. The identity sequence $\NN \rightarrow \NN$, $n\mapsto n$ is $\Ncal$-regular \cite[Prop. 29]{rigo2001}.
The valuation function $\val_\Ncal$ coincides on $L$ with the $\NN$-rational formal power series given by the linear representation $(\lambda, \mu, \rho)$ of dimension $2n$, that is $\lambda\mu(w)\rho = n'$ for all $n'\in\NN$ and $w\in 0^*\rep_\Ncal(n')$: 
\begin{align}\label{lrval}
\lambda&=\begin{pNiceArray}{cccc|ccc}1&0&\cdots&0&0&\cdots&0\end{pNiceArray}&
    \mu : k\in A &\mapsto \begin{pNiceArray}{c|c}
    \Delta_k & V_k\\\hline
    0 & M_\delta
\end{pNiceArray}&
\rho&=\begin{pNiceArray}{ccc|c}0&\cdots&0&\chi_F\end{pNiceArray}^T
\end{align}

\noindent
where, for all $i,j\in Q$ and for all $k\in A$ :
\begin{align*}
M_\delta[i,j] &= \#\{ a \in A | \delta(i,a) = j \}&
V_k[i,j] &= \#\{ a < k | \delta(i,a) = j \}\\
\chi_F[i] &= \begin{cases}
    1 & \mbox{if~} i \in F\\
    0 & \mbox{otherwise}
\end{cases}&
\Delta_k[i,j] &= \begin{cases}
    1 & \mbox{if~} \delta(i,k) = j\\
    0 & \mbox{otherwise}
\end{cases}
\end{align*}
Thus, it is $\NN$-rational.
The linear representation simulates the DFA with $\Delta_k$ and combines the incidence matrix $M_\delta$ with the vector contributions $V_k$ and vector $\rho$ to count the number of accepting path of the same length that precede the input word in radix order.

\subsection{Addition of ANSZ}

The \emph{synchronized addition} $f\dplus g$ between two series $f : A^*\rightarrow\KK$ and $g : A^{\prime *}\rightarrow\KK$ is the series satisfying $(f\dplus g)(\left<w, w'\right>) = f(w) + g(w')$ for all pair of words of the same length $(w,w')\in A^*\times A^{\prime *}$. The synchronized addition of two $\KK$-rational series is $\KK$-rational by running both series at the same time in a block composition of their matrices.

The \emph{support} $\supp(f)$ of a series $f$ on $A$ is the language $A^*\setminus f^{-1}(0)$ of inputs with non-zero image. As the characteristic function of a regular language is $\NN$-rational and as rational series are closed by external multiplication and by Cauchy product, we obtain the following characterization.

\begin{proposition}\label{prop:addrel}
The addition relation $\addrel_\Ncal$ of an ANSZ $\Ncal$ is regular if and only if ${\supp(\val_\Ncal \dplus \val_\Ncal \dplus -\val_\Ncal)}$ is regular.
\end{proposition}

Unfortunately, the above series is a $\ZZ$-series and testing for the regularity of $\ZZ$-series support is known to be undecidable in the general case \cite{berstel2011}. We introduce, through linear recurrence sequence automata, a subfamily of $\ZZ$-series on which regularity can sometimes be effectively obtained.

\section{Linear recurrence sequence automata}\label{sec:LRSA}

An \emph{integer sequence} ${\sequ}\in\ZZ^\NN$ is \emph{linear recurrence} if there exist a \emph{recurrence order} $m$ and \emph{recurrence coefficients} $\alpha_i\in\ZZ$ such that ${\sequ}_{n+m} = \sum_{i=0}^{m-1}\alpha_i \sequ_{n+i}$ for all $n\geqslant 0$. The \emph{recurrence polynomial} of this sequence is $P_{\sequ} = X^m - \sum_{i=0}^{m-1}\alpha_i X^i$.

Let the \emph{zero sequence} be denoted as $(0)$ and let the \emph{shift operator} $\sigma : \ZZ^\NN \rightarrow \ZZ^\NN$ remove the first element of a sequence, \textit{i.e.\/} $(\sigma{\sequ})_n = {\sequ}_{n+1}$ for all ${\sequ}\in\ZZ^\NN$ and $n\in\NN$.
The \emph{vector space} $E_P$ of the linear recurrence sequences of recurrence polynomial  $P = X^m - \sum_{i=0}^{m-1}\alpha_i X^i$ has dimension $m$ and each linear recurrence sequence ${\sequ}\in E_P$ is completely defined by its \emph{initial vector} $V_{\sequ} = \begin{pmatrix} {\sequ}_0,& \ldots,& {\sequ}_{m-1} \end{pmatrix}$. The initial vector of the shifted sequence $\sigma{\sequ} \in E_P$ can be computed using the \emph{companion matrix} $M_P$ of $P$ as $V_{\sigma{\sequ}} = V_{\sequ} M_P$ where
\[
M_P = \left(\begin{smallmatrix}
     0 & 0 & 0 & \dots & 0 & \alpha_0\\
     1 & 0 & 0 & \dots & 0 & \alpha_1\\
     0 & 1 & 0 & \dots & 0 & \alpha_2\\
     \vdots & \vdots & \ddots & \ddots & \vdots & \vdots\\
     \vdots & \vdots & & \ddots & \ddots & \vdots\\
     0 & 0 & 0 & \dots & 1 & \alpha_{m-1}
 \end{smallmatrix}\right)
\]

The construction of the linear representation from Equation~\ref{lrval} uses the incidence matrix $M_\delta$ to keep track of the contribution of each digit in the representation $\rep_\Ncal(n)$ of a natural number $n$. Taking the equivalent generalized Dumont-Thomas numeration system point of view, the matrix $M_\delta$ is indeed the incidence matrix $M_\tau$ of the associated $\tau$ substitution where $M_\tau[i,j] = |\tau(i)|_j$, the number of occurrences of letter $j$ in the word $\tau(i)$. 
The product $\lambda\mu(\rep_\Ncal(n))\rho$ is a linear combination of several $|\tau^k(i)|$. The contribution of a digit to the final value depends on the current state of the automaton and the position of the digit in the number. The sequence of its contributions is a linear recurrence sequence!

Indeed, by the Cayley-Hamilton theorem, the matrix $M_\tau$ verifies the linear recurrence relation given by its monic characteristic polynomial $P$. As $M_\tau^k[i,j]$ counts the number of occurrences of letter $j$ in $\tau^k(i)$, we have
$|\pi(\tau^k(i))| = \sum_{\pi(j)\neq\varepsilon} M_\tau^k[i,j] = \lambda_i M_\tau^k\rho_\pi$
for some vectors $\lambda_i$ and $\rho_\pi$. All the $|\pi(\tau^k(i))|$ are linear recurrence of polynomial $P$. This is also the case for their linear combinations.

From there, we deduce an alternative linear representation $(\lambda,\mu',\rho')$ for the valuation function $\val_\Ncal$, with coefficients in $\ZZ$, where each digit adds its linear recurrence sequence contribution and $\rho'$ simply extracts the first value of the initial vector of the current sequence.
Let the matrix $R$ be such that $R[i,j]=|\pi(\tau^j(i))|$ for all $i,j\in Q$. By construction, $R M_P = M_\delta R$. Now, let $W_k=V_k R$ for all $k\in A$. Use it to construct the new representation:
\begin{align*}
\lambda&=\begin{pNiceArray}{cccc|ccc}1&0&\cdots&0&0&\cdots&0\end{pNiceArray}&
    \mu' : k\in A &\mapsto \begin{pNiceArray}{c|c}
    \Delta_k & W_k\\\hline
    0 & M_P
\end{pNiceArray}&
\rho'&=\begin{pNiceArray}{ccc|cccc}0&\cdots&0&1&0&\cdots&0\end{pNiceArray}^T
\end{align*}
\noindent
where, for all $i,j\in Q$ and for all $k\in A$ :
\begin{align}\label{defwk}
    W_k[i,j] &= \left|\pi\left(\tau^j\left(\tau(i)_0\cdots\tau(i)_{k-1}\right)\right)\right|
\end{align}

In this linear representation, a vector $\lambda\mu'(w)$ memorizes the current state of the DFA together with the initial vector of a current linear recurrence sequence after reading $w$. A transition applies a shift to the current vector and adds the contribution of the current digit. Linear recurrence sequence automata formalize this particular construction.

\subsection{Sequence automata}

A sequence automaton is a partial deterministic finite automaton $(Q, {A}, \delta, q_0, F)$ equipped with \emph{weight} sequences on its edges, given by a partial vector map $\pi$.

\begin{definition}
A \emph{sequence automaton} is a tuple $(Q, A, \delta, q_0, F, \pi)$ where $Q$ is the \emph{finite set of states}, $A$ is the \emph{finite alphabet of symbols}, $\delta : Q\times A\rightarrow Q$ is the \emph{partial} \emph{transition map}, $q_0\in Q$ is the \emph{initial state}, $F\subseteq Q$ is the set of \emph{accepting states} and $\pi : Q\times A \rightarrow \ZZ^\NN$ is the \emph{partial vector map} of the automaton. The transition and the vector map share the same domain.
\end{definition}

The transition map and vector map are inductively extended from symbols to words as follows, for all $q\in Q$, $w\in A^*$ and $a\in A$:
\begin{align*}
\delta(q,\varepsilon) &= q  &\pi(q,\varepsilon) &= \zseq\\
\delta(q, wa) &= \delta(\delta(q, w), a) &\pi(q, wa) &= \sigma\pi(q, w) + \pi(\delta(q, w), a)\quad.
\end{align*}

The intuition behind a \emph{weight} sequence associated with a transition $\delta(q,a)$ is that symbol $a$ contributes for $\pi(q,a)_n$ to the total weight of the word when it appears in position $n$ in the word, counting from right to left --- the same way that digit $3$ counts for $3\times 10^n$ in decimal.

\begin{lemma}
For all $q\in Q$ and $w,w'\in A^*$, when $\delta(q,ww')$ is defined, by construction
$
\pi(q,ww') = \sigma^{|w'|}\pi(q,w) + \pi(\delta(q,w), w')
$.
\end{lemma}

\begin{definition}
    A sequence automaton is \emph{linear recurrence (LRSA)} if every sequence appearing in the vector map is a linear recurrence sequence. The \emph{recurrence polynomial} of a LRSA is the least common multiple of the characteristic polynomials of every linear recurrence sequence in the image of its vector map $\pi$.
\end{definition}

\begin{definition}
    A sequence automaton is \emph{scalar} if there exists a \emph{reference sequence} $\sequ$ of which every sequence appearing in the vector map is a multiple, \textit{i.e.\/} for all $q$ and $a$ such that $\pi(q,a)$ is defined, there exists $\alpha\in\ZZ$ such that $\pi(q,a)=\alpha\sequ$.
\end{definition}

\subsection{Combining sequence automata}

The \emph{scalar product} $\alpha\Acal$ of a sequence automaton $\Acal$ by $\alpha\in\ZZ$ is the same sequence automaton where the vector map $\pi$ is replaced by $\alpha\pi : \sequ \mapsto \alpha\sequ$.

The \emph{addition} $\Acal_1 + \Acal_2$ of two sequence automata $\Acal_i = (Q_i,A_i,\delta,s_i,F_i,\pi_i)$ with $i\in\{1,2\}$ is the product sequence automaton combining the vector maps. Formally, $(Q_1\times Q_2, A_1\times A_2, \delta, (s_1,s_2), F_1\times F_2, \pi)$ where $\delta( (q_1,q_2), (a_1, a_2) ) = (\delta_1(q_1,a_1), \delta_2(q_2,a_2))$ and $\pi( (q_1,q_2), (a_1, a_2) ) = \pi_1(q_1,a_1) + \pi_2(q_2,a_2)$ for all $q_1\in Q_1$, $q_2\in Q_2$, $a_1\in A_1$ and $a_2\in A_2$.

\begin{lemma}
    The recurrence polynomial of a linear combination of LRSA divides the least common multiple of the recurrence polynomials of each automaton of the combination.
\end{lemma}

\subsection{Series of sequence automata and their support}

\begin{definition}
    The \emph{series $s_\Acal$} of a sequence automaton $\Acal=(Q,A,\delta,q_0,F,\pi)$ maps every word $w\in A^*$ to $\pi(q_0,w)[0]$.
\end{definition}

\begin{definition}
The \emph{language $L_0(\Acal)$} of a sequence automaton $\Acal=(Q,A,\delta,q_0,F,\pi)$ is the set of words recognized by the automaton with value $0$, \textit{i.e.\/}
\[
L_0(\mathcal{A}) = \left\{ w\in A^* \mid \delta(q_0, w) \in F \land s_\Acal(w) = 0\right\}
\]
\end{definition}

A simple idea, to transform a sequence automaton into a deterministic automaton recognizing the same language, is to remove the vectors from the transitions and insert them into the states. Alas, the induced flattening deterministic automaton generated this way is potentially infinite.

\begin{definition}
The \emph{flattening} of a sequence automaton $(Q,A,\delta,q_0,F,\pi)$ is the infinite deterministic automaton $(Q', A, \delta', q_0', F')$ where
\begin{align*}
    Q' &= Q\times\ZZ^\NN\\
    \delta'\left( (q,\sequ), a \right) &= \left( \delta(q,a), \sshift\sequ + \pi(q,a) \right)\quad\forall q\in Q, \sequ\in\ZZ^\NN, a\in A\\
    q_0' &= \left( q_0, \zseq  \right)\\
    F' &= \left\{ (q,\sequ) \mid q\in F \land \sequ_0 = 0 \right\}
\end{align*}
\end{definition}

To \emph{trim} a deterministic automaton, one only keeps states accessible from the initial state and co-accessible from the accessible accepting states.

\begin{lemma}\label{lemtrim}
When its trimmed flattening is finite, the language of a sequence automaton is regular.
\end{lemma}

\begin{remark}
    The converse is not true: the language of a LRSA can be regular without the trimmed flattening being finite.
\end{remark}

\begin{lemma}
When the support of its series is regular, the language of a sequence automaton is regular.
\end{lemma}

\section{The Pisot condition}

For LRSA with Pisot recurrence polynomial, the trimmed flattening is always finite. In the scalar case, it is obtained through an adaptation of bounds established by Bruyère and Hansel \cite{bruhan} to study the addition of Bertrand numeration systems. In the general LRSA case the scalar bound can be adapted by combining the polynomial representation from Frougny and Solomyak \cite{frousol} with the change of basis technique from Frougny \cite{frougny2002}.

A \emph{Pisot-Vijayaraghavan number} $\theta$ is an algebraic integer that is the dominant root of its minimal monic polynomial $P(X)$ with integer coefficients, where $P(X)$ is irreducible over $\ZZ$ and admits $n$ complex roots $\theta_1$, \ldots, $\theta_n$, all distinct, satisfying $\theta=\theta_1 > 1 > |\theta_2| \geqslant \ldots \geqslant |\theta_n| > 0$. The set of power sequences of the roots $\left( \theta_i^k \right)_{k\in\NN}$ is a base of the vector space $\lrec_P$.

A linear recurrence sequence $\sequ$ is \emph{Pisot} if its characteristic polynomial is the minimal polynomial $P(X)$ of a Pisot-Vijayaraghavan number. It is \emph{Pisot, up to shift,} if $\sshift^k\sequ$ is Pisot for some $k\geqslant 0$, that is if its polynomial is $X^k P(X)$. By extension, such a polynomial is said to be \emph{Pisot, up to shift}.

\subsection{Scalar Pisot LRSA}

\begin{theorem}\label{thm:scalar}
The language of a scalar LRSA with a reference sequence that is Pisot, up to shift, is regular.
\end{theorem}

\begin{proof}
Let $(Q, {A}, \delta, q_0, F, \pi)$ be a scalar LRSA with reference sequence $\sequ$. Assume that $\sequ$ is Pisot, up to shift, and let $X^k P(X)$ be the characteristic polynomial of $\sequ$. Let  $\theta=\theta_1$, \ldots, $\theta_n$ be the roots of $P$ satisfying $\theta=\theta_1 > 1 > |\theta_2| \geqslant \ldots \geqslant |\theta_n| > 0$.
 Let $\omega : Q\times{A} \rightarrow \ZZ$ retrieve the scalar coefficient of the weight so that $\pi(q,a) = \omega(q,a)\sequ$ for all $q\in Q$ and $a\in{A}$. Extend $\omega$ inductively to $\omega : Q\times{A}^* \rightarrow \ZZ^*$, where $(\ZZ^*,\cdot)$ denotes the monoid of words on $\ZZ$, by $\omega(q,\varepsilon)=\varepsilon$ and $\omega(q,wa)=\omega(q,w)\cdot\omega(\delta(q,w),a)$ for all $q\in Q$, $w\in{A}^*$ and $a\in{A}$.

 As the number of transitions of the sequence automaton is finite, the scalar coefficients are bounded. Let $C = \max_{(q,a)\in Q\times{A}} |\omega(q,a)|$. The coefficients all lie in the interval $\mathcal{C} = \{-C, \ldots, C\}$.

 Let $\pi_\sequ : \ZZ^*\rightarrow \ZZ$ assign a weight to every finite sequence of integers according to $\sequ$ by ${\pi_\sequ(a_m\cdots a_0) = \sum_{i=0}^{m} a_i \sequ_i}$ for all $a_m\cdots a_0 \in \ZZ^*$. By construction, $(\pi(q, w))_i = \pi_\sequ\left(\omega(q, w)0^i\right)$, or equivalently $(\pi(q, w))_i = \pi_{\sshift^i \sequ}\left(\omega(q, w)\right)$, for all $i\in\NN$, $q\in Q$ and $w\in{A}^*$.

 The key argument, from \cite{bruhan}, is to approximate $\pi_\sequ$ using powers of $\theta$. As $\sigma^k\sequ \in \lrec_P$ there exist $\gamma_1, \ldots, \gamma_n\in \CC$ such that $\sequ_{m+k} = \sum_{i=1}^n \gamma_i\theta_i^m$ for all $m\in\NN$. Moreover $\gamma = \gamma_1$ is a real number. Let $\pi_\theta : \ZZ^*\rightarrow \RR$ assign a weight to every finite sequence $a_m\cdots a_0\in\ZZ^*$  by ${\pi_\theta(a_m\cdots a_0) = \gamma\sum_{i=k}^{i=m} a_i \theta^{i-k}}$, ignoring the first $k$ values.
First, let us bound the approximation error by $C$ times a constant $K$. Let $a_m\cdots a_0 \in \mathcal{C}^*$,
\begin{align*}
    \pi_\sequ(a_m\cdots a_0) - \pi_\theta(a_m\cdots a_0) &= \sum_{i=0}^{k-1} a_i\sequ_i - \sum_{i=k}^m a_i\sum_{j=2}^n \gamma_j\theta_j^{i-k}\\
    &= \sum_{i=0}^{k-1} a_i\sequ_i - \sum_{j=2}^n\gamma_j\sum_{i=k}^m a_i\theta_j^{i-k}\\
    \mbox{thus}\quad\left| \pi_\sequ(a_m\cdots a_0) - \pi_\theta(a_m\cdots a_0) \right| &\leqslant C\left( \sum_{i=0}^{k-1} |\sequ_i| + \sum_{j=2}^n|\gamma_j|\sum_{i=k}^m |\theta_j|^{i-k} \right)\\
    \left| \pi_\sequ(a_m\cdots a_0) - \pi_\theta(a_m\cdots a_0) \right| &\leqslant C\underbrace{\left( \sum_{i=0}^{k-1} |\sequ_i| + \sum_{j=2}^n\frac{|\gamma_j|}{1-|\theta_j|}\right)}_{=K}
\end{align*}
\noindent
Consider now an accepting path of the automaton. Let $w, w'\in{A}^*$ such that $\pi(q_0,ww')$ is defined and verifies $(\pi(q_0,ww'))_0 = 0$. Let $m=|w'|$. Let $q=\delta(q_0,w)$. As $\pi_\sequ(\omega(q_0, ww')) = 0$, we have
\begin{alignat*}{2}
    -CK&\leqslant& \pi_\theta(\omega(q_0, ww')) &\leqslant CK\\
    \mbox{thus}\quad -CK&\leqslant& \pi_\theta(\omega(q_0, w)0^m) + \pi_\theta(\omega(q, w')) &\leqslant CK\\
    -C\left(K + \frac{\gamma\theta}{\theta-1} \theta^m\right)&\leqslant& \pi_\theta(\omega(q_0, w))\theta^m &\leqslant C\left(K + \frac{\gamma\theta}{\theta-1} \theta^m\right)\\
    \mbox{then diving by $\theta^m$,}&& \left| \pi_\theta(\omega(q_0, w)) \right| &\leqslant C\left(K + \frac{\gamma\theta}{\theta-1} \right)
\end{alignat*}
\noindent
And finally we bound $\pi_\sequ\left(\omega(q_0,w)0^i\right)$ for all $i$ to prepare the bound for $V_{\pi(q_0,w)}$:
\begin{align*}
    \left| \pi_\sequ\left(\omega(q_0,w)0^i\right) \right| &\leqslant \left| \pi_\sequ\left(\omega(q_0,w)0^i\right) - \pi_\theta\left(\omega(q_0,w)0^i\right) \right| + \left| \pi_\theta\left(\omega(q_0,w)0^i\right) \right|\\
    &\leqslant C\left( K + \left(K + \frac{\gamma\theta}{\theta-1}\right)\theta^i \right)
\end{align*}
$V_{\pi(q_0,w)} = (\pi_\sequ\left(\omega(q_0,w)\right), \pi_{\sshift\sequ}\left(\omega(q_0,w)\right), \cdots, \pi_{\sshift^{m-1}\sequ}\left(\omega(q_0,w)\right))$ is bounded thus the trimmed flattening of the sequence automaton is finite.\qed
\end{proof}

Note that the whole process is constructive. Given a scalar LRSA, one can effectively construct the bounded flattening by encoding sequences as initial vectors and by using the companion matrix to compute transitions.

\subsection{Pisot LRSA}

Consider a sequence $\sequ\in\lrec_P$ verifying the recurrence relation given by some monic polynomial $P$ of degree $m$. The set of shifted sequences $\{ \sshift^i\sequ \}_{i\in\NN}$ generates a subspace of $\lrec_P$. It generates the whole space $\lrec_P$ if and only if the following associated Hankel matrix is invertible:
\[
H_\sequ = \left(\begin{smallmatrix}
    \sequ_0 & \sequ_1 & \hdots & \sequ_{m-1} \\
    \sequ_1 & \sequ_2 & \hdots & \sequ_{m} \\
    \vdots & \vdots & \ddots & \vdots \\
    \sequ_{m-1} & \sequ_m & \hdots & \sequ_{2m-2} \\
 \end{smallmatrix}\right)
\]
In this case, for every sequence $\seqv\in\lrec_P$, we have $\seqv = \sum_{i=0}^{m-1} \alpha_i \sshift^i\sequ$ where ${(\alpha_0,\ldots,\alpha_{m-1}) = V_\seqv H_\sequ^{-1}}$. By letting $R = \sum_{i=0}^{m-1}\alpha_i X^i$, one can write $\seqv = R(\sshift)(\sequ)$ and see every sequence of $\lrec_P$ as a polynomial in $\sequ$ with rational coefficients. The choice of the sequence $\sequ$ is free. One might for example choose the sequence of initial vector $(0,\ldots,0,1)$ for which $H_\sequ$ is always invertible, as it is anti-triangular, with a determinant of value $1$ or $-1$, ensuring integer coefficients for the polynomials.

\begin{theorem}\label{thm:lrsa}
The language of a LRSA with a recurrence polynomial that is Pisot, up to shift, is regular.
\end{theorem}

\begin{proof}
Let $\mathcal{A}$ be a LRSA $(Q, {A}, \delta, q_0, F, \pi)$ with recurrence polynomial $P$ of degree $m$. Assume that $P$ is Pisot, up to shift, and let $\sequ\in\lrec_P$ be a sequence such that the set of its shifted sequences $\{ \sshift^i\sequ \}_{i\in\NN}$ generates $\lrec_P$.
Let $\omega : Q\times{A} \rightarrow \QQ[X]$ retrieve a polynomial representation, of degree at most $m-1$, of the weight so that $\pi(q,a) = \omega(q,a)(\sshift)(\sequ)$ for all $q\in Q$ and $a\in{A}$. Extend $\omega$ inductively to $\omega : Q\times{A}^* \rightarrow \QQ[X]$ by $\omega(q,\varepsilon)=0$ and $\omega(q,wa)=\omega(q,w)X+\omega(\delta(q,w),a)$ for all $q\in Q$, $w\in{A}^*$ and $a\in{A}$. By construction, $(\pi(q, w))_i = \left(\omega(q, w) X^i\right)(\sshift)(\sequ)[0]$ for all $i\in\NN$, $q\in Q$ and $w\in{A}^*$.

As the number of transitions of the automaton is finite, the polynomial representations have bounded coefficients. Let $C = \max_{\sum_{i=0}^{m-1}\alpha_i X^i\in \pi(Q,{A})} |\alpha_i|$. The coefficients all lie in the interval $\mathcal{C} = \{-C, \ldots, C\}$. Moreover, the coefficients of $(\pi(q, w))_i$ all lie in the interval $m\mathcal{C} = \{-mC, \ldots, mC\}$ for all $i\in\NN$, $q\in Q$ and $w\in{A}^*$.

When the automaton $\mathcal{A}$ admits a loop $\delta(q_0,a)=q_0$ with weight $\pi(q_0,a)=0$ for some symbol $a$, one can construct a scalar LRSA $\mathcal{B}$ with reference sequence $\sequ$ so that $\mathcal{A}$ and $\mathcal{B}$ recognize the same language. To do so, $\mathcal{B}$ anticipates, with a sliding windows on tuples of letters, the future coefficients of the polynomials: instead of associating the weight $\sum_{i=0}^{m-1}\alpha_i X^i$ to a transition $\delta(q,b)$, the weights $\alpha_i$ for $i>0$ are added to the weights of the $m-1$ previous symbols. By Theorem~\ref{thm:scalar}, the language of $\mathcal{B}$ is regular.

If the sequence automaton admits no loop of weight zero on its initial state, consider adding to the automaton $\mathcal{A}$ a new padding symbol $\#\not\in{A}$ and a loop $\delta(q_0,\#)=q_0$ of weight $\pi(q_0,\#)=0$ to obtain a sequence automaton $\mathcal{A}^\#$. By construction,
${L_0(\mathcal{A}) = \{ u\in{A}^* | \exists k<m, \#^k u\in L_0(\mathcal{A}^\# ) \}}$. By previous argument, the language of $\mathcal{A}^\#$ is regular,  and so is the language of $\mathcal{A}$.\qed
\end{proof}

Note that the whole process is still constructive and that the proof provides an effective method to transform the sequence automaton into an automaton of a single sequence. 

\section{Applications}

Using the expressivity of LRSA, we can answer our initial question in the Pisot case and handle several other constructions.

\medskip
\par\noindent\textbf{Addition of Pisot ANSZ~}
Following Equation~\ref{defwk}, the \emph{addressing sequence automaton} $\Acal_\tau$ associated with the fixpoint $\pi(\tau^\omega(a))$ of a generalized Dumont-Thomas numeration system is obtained by adding a vector map $\pi$ to $\Ncal_\tau$ as $\pi(b,k) = \left(\left|\pi\left(\tau^n\left(\tau(b)_0\cdots\tau(b)_{k-1}\right)\right)\right|\right)_{n\in\NN}$ for all $b\in A$ and $k\in B$. The series of this sequence automaton is the valuation function of the numeration system. It is a LRSA whose recurrence polynomial divides the characteristic polynomial of the incidence matrix of $\tau$ (see for example Fig.~\ref{fig:seqpsi}). By combining Proposition~\ref{prop:addrel} with Theorem~\ref{thm:lrsa}, Theorem~\ref{addDT} is proved.

\medskip
\par\noindent\textbf{Conversion between ANSZ~}
Let $\Ncal$ and $\Ncal'$ be two ANSZ whose valuation addressing sequence automata share a same Pisot recurrence polynomial, up to shift. Both ANSZ are equivalent: the series $\val_\Ncal \dplus -\val_{\Ncal'}$ has a regular support so the conversion between both ANSZ is regular. In particular, the equivalence to the canonical Bertrand numeration system \cite{bruhan} is effectively constructed.

\medskip
\par\noindent\textbf{Parikh vectors of fixpoint prefixes~} Using the same techniques, the Parikh vectors of prefixes of the fixpoint $\pi(\tau^\omega(a))$ of a generalized Dumont-Thomas numeration system can be obtained by constructing, for each letter $b\in B$, the erasing coding $\pi_b$ that erases every letter but $b$ and computing the series $\val_\Ncal \dplus -\val_{\Ncal_b}$ where $\Ncal$ is the generalized Dumont-Thomas numeration system of $\pi(\tau^\omega(a))$ and $\Ncal_b$ the one of $(\pi_b\circ\pi)(\tau^\omega(a)))$.

\section{Practical considerations}

We have developed a prototype tool, \verb+licofage+ \cite{licofage}. Given a substitution, \verb+licofage+ can produce its addressing sequence automaton and compute linear combinations of their series. In the Pisot case, or when given a manually selected bound, it produces the trimmed flattening deterministic finite automaton that computes a relation defined using linear combinations of sequence automata series (addition, conversion, \textit{etc}). The generated output can be seamlessly integrated into Walnut \cite{walnuta}. A companion notebook is available on arXiV to showcase the capabilities of the proposed method, accessible via the following URL: \url{https://arxiv.org/src/2406.09868/anc}.

\newpage
\pagestyle{empty}
\appendix
\section{Additional proofs}

\subsection{Proof of Proposition~\ref{progdtns}}

\begin{proof}
Let $\Ncal$ be a ANSZ $(L,A,<,0)$ and let $\Acal=(Q,A,\delta,q_0,F)$ be a minimized trimmed DFA recognizing $L$. Let $k$ be the maximum number of transitions from a state of $\Acal$.
Let $\Acal'$ be the DFA $(Q,\{0,\ldots,k-1\},\delta',q_0,F)$ where $\delta'(q,i)=q'$ is defined for all $q,q'\in Q$ and $i<k$ if $q$ admits at least $i+1$ transitions and $\delta(q,a)=q'$ where $a$ is the label of the $i$th transition starting from $q$ according to $<$.
The canonical ANSZ associated to $\Ncal$ is $(L(\Acal'),\{0,\ldots,k-1\},<,0)$. It is equivalent to $\Ncal$ by a simple state-dependent rewriting of transitions. Moreover, it is the generalized Dumont-Thomas numeration system of associated to $\pi(\tau^\omega(a))$ where $a=q_0$, $\tau : q \mapsto \delta(q,0)\cdots\delta(q,k_q)$ where $k_q$ is the number of transitions starting from $q$ and $\pi(q)=q$ if $q\in F$ and $\pi(q)=0$ otherwise.\qed
\end{proof}

\subsection{Proof of Lemma~\ref{lemtrim}}
\begin{proof}
    The language of the sequence automaton is equal to the language of its flattening. Indeed, for every $w\in A^*$ such that $\delta(q_0,w)$ is defined, we have ${\delta'(q_0',w) = (\delta(q_0,w),\pi(q_0,w))}$. Thus $\delta'(q_0',w) \in F'$ if and only if $\delta(q_0,w)\in F$ and $\left(\pi(q_0,w)\right)_0 = 0$.\qed
\end{proof}

\clearpage
\section{Extra material}
\subsection{$\ZZ$-rational valuation of $\Ncal_\tau$}

Following the paper, we provide the explicit linear presentations for the valuation obtained for ${\tau : a\mapsto aab,\; b\mapsto cab,\; c\mapsto a}$ using both techniques. One can easily check that both linear representation are equivalent, they both compute the same series.

\subsubsection{Following Section~\ref{sec:val}}
Let us recall the formulas:
\begin{align*}
\lambda&=\begin{pNiceArray}{cccc|ccc}1&0&\cdots&0&0&\cdots&0\end{pNiceArray}&
    \mu : k\in A &\mapsto \begin{pNiceArray}{c|c}
    \Delta_k & V_k\\\hline
    0 & M_\delta
\end{pNiceArray}&
\rho&=\begin{pNiceArray}{ccc|c}0&\cdots&0&\chi_F\end{pNiceArray}^T
\end{align*}

\noindent
where, for all $i,j\in Q$ and for all $k\in A$ :
\begin{align*}
M_\delta[i,j] &= \#\{ a \in A | \delta(i,a) = j \}&
V_k[i,j] &= \#\{ a < k | \delta(i,a) = j \}\\
\chi_F[i] &= \begin{cases}
    1 & \mbox{if~} i \in F\\
    0 & \mbox{otherwise}
\end{cases}&
\Delta_k[i,j] &= \begin{cases}
    1 & \mbox{if~} \delta(i,k) = j\\
    0 & \mbox{otherwise}
\end{cases}
\end{align*}
In the particular case of $\tau$ we have $A=\{0,1,2\}$ and:
\begin{align*}
\lambda &= \begin{pmatrix}
   1&0&0&0&0&0 
\end{pmatrix}&
M_\delta &= \begin{pmatrix}
    2&1&0\\
    1&1&1\\
    1&0&0
\end{pmatrix}&
\rho &= \begin{pmatrix}
    0&0&0&1&1&1
\end{pmatrix}^T\\
\Delta_0 &= \begin{pmatrix}
    1&0&0\\
    0&0&1\\
    1&0&0
\end{pmatrix}&
\Delta_1 &= \begin{pmatrix}
    1&0&0\\
    1&0&0\\
    0&0&0
\end{pmatrix}&
\Delta_2 &= \begin{pmatrix}
    0&1&0\\
    0&1&0\\
    0&0&0
\end{pmatrix}\\
V_0 &= \begin{pmatrix}
    0&0&0\\
    0&0&0\\
    0&0&0
\end{pmatrix}&
V_1 &= \begin{pmatrix}
    1&0&0\\
    0&0&1\\
    1&0&0
\end{pmatrix}&
V_2 &= \begin{pmatrix}
    2&0&0\\
    1&0&1\\
    1&0&0
\end{pmatrix}
\end{align*}

\subsubsection{Following Section~\ref{sec:LRSA}} Let us recall the formulas:
\begin{align*}
\lambda&=\begin{pNiceArray}{cccc|ccc}1&0&\cdots&0&0&\cdots&0\end{pNiceArray}&
    \mu' : k\in A &\mapsto \begin{pNiceArray}{c|c}
    \Delta_k & W_k\\\hline
    0 & M_P
\end{pNiceArray}&
\rho'&=\begin{pNiceArray}{ccc|cccc}0&\cdots&0&1&0&\cdots&0\end{pNiceArray}^T
\end{align*}

\noindent
where, for all $i,j\in Q$ and for all $k\in A$ :
\begin{align*}\label{defwk}
    W_k[i,j] &= \left|\pi\left(\tau^j\left(\tau(i)_0\cdots\tau(i)_{k-1}\right)\right)\right|
\end{align*}
In the particular case of $\tau$ we have $A=\{0,1,2\}$ and:
\begin{align*}
\lambda &= \begin{pmatrix}
   1&0&0&0&0&0 
\end{pmatrix}&
M_P &= \begin{pmatrix}
    0&0&1\\
    1&0&-1\\
    0&1&3
\end{pmatrix}&
\rho' &= \begin{pmatrix}
    0&0&0&1&0&0
\end{pmatrix}^T\\
W_0 &= \begin{pmatrix}
    0&0&0\\
    0&0&0\\
    0&0&0
\end{pmatrix}&
W_1 &= \begin{pmatrix}
    1&3&9\\
    1&1&3\\
    1&3&9
\end{pmatrix}&
W_2 &= \begin{pmatrix}
    2&6&18\\
    2&4&12\\
    1&3&9
\end{pmatrix}
\end{align*}

\subsubsection{Equivalence of representations} The key argument is to define the matrix $R$ such that $R[i,j] = |\pi(\tau^j(i))|$ for all $i,j\in Q$ and verify that $W_k=V_k R$ and $R M_P = M_\delta R$ for all $k\in A$. Here:
\begin{align*}
    R &= \begin{pmatrix}
        1&3&9\\
        1&3&7\\
        1&1&3
    \end{pmatrix}
\end{align*}
From there, on gets that $\lambda\mu'(w) = \lambda\mu(w)
\begin{pNiceArray}{c|c}
    \operatorname{Id} & 0\\\hline
    0 & R 
\end{pNiceArray}$ and thus $\lambda\mu'(w)\rho' = \lambda\mu(w)\rho$ for all $w\in A^*$ as $\rho = \begin{pNiceArray}{c|c}
    \operatorname{Id} & 0\\\hline
    0 & R
\end{pNiceArray}\rho'$.
\
\end{document}